\documentclass[aoas,preprint]{imsart}

\usepackage{amsthm,amsmath,amsfonts}
\RequirePackage[colorlinks,citecolor=blue,urlcolor=blue]{hyperref}
\usepackage[round]{natbib}
\usepackage[english]{babel}
\usepackage{latexsym}
\usepackage{subfigure}
\usepackage{epsfig}
\usepackage{color,soul}
\usepackage{moreverb}
\usepackage{mathtools}
\usepackage{tikz, ulem,color}
\usepackage{enumerate,linegoal}
\usepackage{calc}

\usetikzlibrary{matrix}

\startlocaldefs
\newcommand{\expect}{\operatorname{E}}
\newcommand{\calR}{{\cal{R}}}
\newcommand{\R}{{\mathbb{R}}}
\newtheorem{remark}{Remark}[section]
\newcommand{\rd}{\mathrm{d}}
\newcommand{\rds}{\,\rd}
\newcommand\independent{\protect\mathpalette{\protect\independenT}{\perp}}
\def\independenT#1#2{\mathrel{\rlap{$#1#2$}\mkern2mu{#1#2}}}
\newtheorem{assumption}{Assumption}
\newtheorem{proposition}{Proposition}
\endlocaldefs

\begin{document}

\begin{frontmatter}

\title{A novel quantile-based decomposition of the indirect effect in mediation analysis with an application to infant mortality in the US population}
\runtitle{Quantile-based mediation analysis}

\begin{aug}
\author{\fnms{Marco} \snm{Geraci}\corref{}\thanksref{t1,m1}\ead[label=e1]{geraci@mailbox.sc.edu}}
\and
\author{\fnms{Alessandra} \snm{Mattei}\thanksref{t2}\ead[label=e2]{mattei@disia.unifi.it}}

\thankstext{t1}{Corresponding author: Marco Geraci, Department of Epidemiology and Biostatistics, Arnold School of Public Health, University of South Carolina, 915 Greene Street, Columbia SC 29209, USA. \printead{e1}}

\thankstext{m1}{Supported in part by funding received from the National Institutes of Health -- National Institute of Child Health and Human Development (Grant Number: 1R03HD084807-01A1).}

\runauthor{Geraci and Mattei}

\affiliation{University of South Carolina\thanksmark{t1} and University of Florence\thanksmark{t2}}



\end{aug}

\begin{abstract}
\quad In mediation analysis, the effect of an exposure (or treatment) on an outcome variable is decomposed into two components: a direct effect, which pertains to an immediate influence of the exposure on the outcome, and an indirect effect, which the exposure exerts on the outcome \textit{through} a third variable called mediator. Our motivating example concerns the relationship between maternal smoking (the exposure, $X$), birthweight (the mediator, $M$), and infant mortality (the outcome, $Y$), which has attracted the interest of epidemiologists and statisticians for many years. We introduce new causal estimands, named $u$-specific direct and indirect effects, which describe the direct and indirect effects of the exposure on the outcome at a specific quantile $u$ of the mediator, $0 < u < 1$. Under sequential ignorability we derive an interesting and novel decomposition of $u$-specific indirect effects. The components of this decomposition have a straightforward interpretation and can provide new insights into the complexity of the mechanisms underlying the indirect effect. We illustrate the proposed methods using data on infant mortality in the US population. We provide analytical evidence that supports the hypothesis that the risk of sudden infant death syndrome is not predicted by changes in the birthweight distribution.
\end{abstract}



\end{frontmatter}


\section{Introduction}
\label{sec:1}
In epidemiological research, assessing the effect of an exposure ($X$) on the outcome of interest ($Y$) is often done in relation to a third variable ($M$), called intermediate variable or mediator, which is suspected to lie on the causal pathway between the exposure and the outcome under investigation. Mediation analysis is a popular statistical approach to this kind of problems and occupies an important place in observational studies. The goal of mediation analysis is to tease out direct and indirect effects of an exposure on the outcome.

We consider the potential outcome approach to causal inference as proposed by \cite{Rubin1974, Rubin1990} and define channeled and un-channeled effects by means of natural direct and indirect effects \citep{RobinsGreenland1992, Pearl2001}. Natural direct and indirect effects describe what would happen to the relationship between exposure and outcome when the mediator is intervened upon.

In mediation analysis, a sequential ignorability assumption is usually invoked, which implies that conditional on pretreatment covariates, say $W$, there is no unmeasured confounding of the treatment-mediator, treatment-outcome and mediator-outcome relationships. Under this assumption, the analysis of the data revolves around the conditional distributions of $M|X,W$ and of $Y|X,M,W$. Most of the literature on mediation analysis concerns location-shift effects, which is tantamount to the application of (generalized) linear regression to model the relationship between $X$ and the conditional mean of $M$, or between $(X,M)$ and the conditional mean of $Y$, given $W$. A few exceptions include \cite{imai_2010a} and \cite{shen_etal_2014}, who proposed assessing the effects of the exposure on the conditional quantiles of the outcome, and \cite{dominici_etal}, who considered conditional quantiles of the mediator when estimating direct and indirect effects.

In this paper, our primary interest does not involve how the distribution of the outcome is summarized nor this issue is relevant to the forthcoming discussion. Although we choose the expected value as parameter of interest, our results extend to any outcome model, including \citeauthor{imai_2010a}'s (\citeyear{imai_2010a}) model for `quantile causal mediation effects' and \citeauthor{shen_etal_2014}'s (\citeyear{shen_etal_2014}) `quantile outcome model'. In contrast, our goal is to provide a break down of the natural indirect effects by using an alternative formulation based on the quantiles of the mediator model. In this regard, our proposal is novel and, to the best of our knowledge, there are no published studies comparable to ours. Specifically, we first order the potential values of the mediator under alternative treatment conditions according to their cumulative probabilities, say $u$, where $u \in (0,1)$. Then we show that the (overall) average indirect effect can be written as the average of what we labelled \textit{$u$-specific indirect effects}, that is, the expected indirect effects of the treatment on the outcome given the quantiles of the mediator. Similarly, we show that the (overall) average total effect and the (overall) average natural direct effect can be written as averages of conditional expected total and direct effects given the quantiles of the mediator, which we refer to as \textit{$u$-specific total and direct effects}, respectively. In this regard, our study offers another element of novelty as compared to \citeauthor{dominici_etal}'s (\citeyear{dominici_etal}) paper, namely a direct link between average causal effects and quantile-specific effects.

We make the assumption of sequential ignorability \citep[e.g.,][]{Pearl2001,dominici_etal,VanderWeeleVansteelandt2009,imai_2010a}, which allows us to identify and estimate average natural direct and indirect effects from the observed data. Nevertheless our focus is on the $u$-specific effects. Under sequential ignorability, we investigate the information provided by the data on $u$-specific indirect effects, and, by using the formulation of the average natural indirect effect as average of $u$-specific indirect effects, we obtain a novel decomposition of the indirect effect, which provides additional valuable and easy-to-interpret information about the mediating process. The crucial quantities involved in the decomposition are the conditional quantile function of the mediator given the exposure, the density of the mediator's distribution, and the sensitivity of the outcome to changes in the mediator's distribution, all of which we propose to estimate using a distribution-free approach.

We illustrate our approach, including modelling and inferential strategies, using a dataset with several millions of observations on birthweights and infant mortality in the US population.

\section{Smoking, birthweight and infant mortality}
\label{sec:2}
Our motivating example is represented by maternal smoking during pregnancy (the exposure), birthweight (the mediator) and infant mortality (the primary outcome). We consider data on livebirths and infant deaths (i.e., deaths of children younger than one year of age) in white singletons born in the United States (US) between 2001 and 2005 obtained from the National Center for Health Statistics (NCHS). These data contain birth-cohort linked information on cause of death, age at death, birthweight, maternal smoking (including approximate number of cigarettes smoked daily), and other medical and sociodemographic characteristics systematically recorded on US birth certificates.

Tobacco smoking during pregnancy increases the risk of a number of adverse outcomes, including miscarriage, placental abruption, preterm delivery, and sudden infant death syndrome (SIDS). It is also known that the effect of smoking on birthweight is to shift the entire distribution to the left, thus increasing the risk of low birthweight (LBW) ($< 2500$ gr). However, the shift is not uniform across birthweight quantiles but stronger on lower quantiles and weaker on upper quantiles \citep{abrevaya_2001, koenker_hallock, geraci_2016}. Finally, numerous studies have linked infant mortality to birthweight \citep{wilcox_1983}. Mortality is highest in infants born very small and it decreases monotonically with increasing birthweight (except in macrosomic infants for whom mortality tends to rise). The International Classification of Diseases (ICD), Tenth Edition \citep{icd10}, lists `Extremely low birthweight newborn' (P07.0) ($< 1000$ gr) and `Other low birth weight newborn' (P07.1) ($1000$--$2500$ gr) as causes of morbidity or additional care in infants. LBW may also be related to diseases in later life. Using a life course perspective \citep{destavola}, epidemiology studies are now investigating the association of LBW with adverse health outcomes in childhood  \cite[e.g., in relation to cancer risks,][]{koifman_2008,birch_2010} as well as in adulthood (the `fetal origins' hypothesis) \citep{godfrey_2000,huxley_2002,huxley_2004}.

What needs to be clarified, yet, is the mediating role of birthweight in the association between smoking and infant mortality risk. \cite{wilcox_2001} speculated that the effect of smoking is to shift the birthweight distribution and the mortality curve to the left, uniformly at all birthweight quantiles. This would imply that the impact of smoking on mortality is independent of its effect on birthweight, i.e., that the indirect effect is null. In a related commentary \citep{picciotto_2001}, it was noted that the very results reported by \cite{wilcox_2001} were actually pointing in the opposite direction, i.e. towards a non-uniform effect of smoking on mortality. However, no analytical evidence was provided to support either claims.

It has also been suggested that there might be common causes of LBW and mortality (e.g., birth defects), which usually remain unobserved to the analyst \citep{hernandez_2006}. In other words, birthweight \textit{per se} is not a causative factor, but rather an endpoint of prenatal biological mechanisms. Gestational age and birthweight for gestational age are often preferred for predicting health outcomes. However, birthweight still has a role in predicting health outcomes as a proxy of unmeasured processes \citep{picciotto_2001}. Moreover, birthweight is cheap to measure, which is especially advantageous in low and middle income countries, and is measured more accurately than gestational age, since timing of ovulation and conception are often uncertain. As a consequence, missing and misclassified data are usually higher in proportion for gestational age than for birthweight, and, even worse, they are more common in women that are more likely to give birth to infants at higher risk of mortality \citep{wilcox_2001}.

These facts make smoking, birthweight, and infant mortality a perfect candidate for a mediation analysis. However, due to the reasons discussed above, statistical methods that deal only with the location-shift effect of the exposure (smoking) on the mediator (birthweight) might be missing out on important clues, which would help understand the mediating process.

\section{Methods}
\label{sec:3}
\subsection{Potential outcomes and causal estimands}
\label{sec:3.1}
We consider the simple case in which the exposure variable, $X$, is dichotomous, namely $X=0$ in the unexposed population and $X=1$ in the exposed population. The outcome $Y$ can be either discrete or continuous. We focus on an absolutely continuous intermediate variable $M$. We denote by $W$ a vector of observed pretreatment variables.

Under the stable unit treatment value assumption (SUTVA) \citep{Rubin1990}, which rules out the presence of different versions of each treatment level and interference between units, we can define two potential outcomes for each post-treatment variable. Let $M(x)$ and $Y(x)$ denote respectively the potential outcomes of $M$ and $Y$ if treatment $X$ were set, possibly contrary to fact, to the value $x$, $x=0,1$. Under an appropriate version of SUTVA \citep[see, for example,][]{MatteiMealli2011}, we can define the following potential outcomes:
\begin{itemize}
\item $Y(x,m)$, which would be the value of the outcome $Y$ if the treatment were set to the level $x$ and the mediator $M$ were set to a specific prefixed value, $m$; and
\item $Y(x, M(x^\ast))$, which would be the value of the outcome $Y$ if the treatment were set to the level $x$ and the mediator $M$ were set to the value it would have taken if the treatment had been set to an alternative level, $x^\ast$. Under the composition assumption \citep{VanderWeele2015},  $Y(x, M(x)) = Y(x)$.
\end{itemize}

We concentrate on causal effects defined by the differences of potential outcomes. Therefore the average total causal effect (ACE) of the treatment $X$ on the outcome $Y$ is defined as the mean difference between potential outcomes:
\begin{align}\label{eq:1}
ACE & = \expect\left(Y(1)-Y(0)\right)\\
\nonumber &= \expect_W\left\{\expect\left(Y(1)-Y(0)\mid W=w\right)  \right\} = \expect_W\left(ACE_w \right),
\end{align}
where the outermost expectation is over the distribution of the pretreatment covariates and $ACE_w=\expect\left(Y(1)-Y(0)\mid W=w\right)$ is the average total causal effect conditional on covariates at level $W=w$.

In the presence of an intermediate variable $M$ it could be of interest decomposing the total effect of the exposure $X$ on the outcome $Y$  into a channeled (indirect) effect  mediated through $M$, and an un-channeled (direct) effect, that is, an effect not mediated through $M$.  Here we consider natural indirect (NIE) and direct (NDE) effects. The conditional average NIE and NDE given $W=w$ are defined by
\begin{align}
\label{eq:2} NIE_{x \mid w } & = \expect\left(Y(x, M(1))  - Y(x, M(0)) \mid W=w\right),\\
\label{eq:3} NDE_{x^\ast \mid w} & = \expect\left(Y(1, M(x^\ast)) - Y(0, M(x^\ast)) \mid W=w\right),
\end{align}
for $x,x^\ast=0,1$ \citep[see, e.g,][]{RobinsGreenland1992, Pearl2001}. Then, the average total effect conditional on $W = w$ can be decomposed into the sum of a natural indirect effect and a natural direct effect as follows:
\begin{equation*}
ACE_w = NIE_{0\mid w} + NDE_{1\mid w}= NID_{1\mid w} + NDE_{0\mid w}.
\end{equation*}
For convenience, all the effects are defined conditional on covariates, but it is worth noting that $NIE_{x}=\expect_W\left\{NIE_{x \mid w }\right\}$ and $NDE_{x^\ast}=\expect_W\left\{NDE_{x^\ast \mid w }\right\}$ and that the following decomposition of the total effect holds: $ACE = NIE_{0} + NDE_{1}= NIE_{1} + NDE_{0}$.

Let's define the function
\begin{equation}\label{eq:4}
\calR_{x,x^\ast}(m \mid w) = \expect\left(Y(x,m) \mid M(x^\ast)=m, W=w\right),
\end{equation}
for $m \in \R$, $x, x^\ast=0,1$. The interpretation of $\calR_{x,x^\ast}(m \mid w)$ is straightforward. For example, if we refer to the NCHS birthweight study introduced previously, then $\calR_{1,x^\ast}(m \mid w)$ describes the expected infant mortality that would result if infants in the subpopulation $W = w$ were exposed to smoking, as a function of the potential birthweights that could be observed if mothers either smoked ($x^\ast = 1$) or did not smoke ($x^\ast = 0$) during pregnancy.

By applying the law of iterated expectations we obtain
\begin{align*}
\lefteqn{\expect\left(Y(x, M(x^\ast)) \mid W=w\right) }\\&=\int_{\mathbb{R}}\expect\left(Y(x,m) \mid M(x^\ast)=m, W=w\right)\rds F_{M(x^\ast)\mid W=w}(m).
\end{align*}
Then, we can rewrite the natural indirect and direct effects  as follows
\begin{align}
\label{eq:5} NIE_{x\mid w} = & \int_{\mathbb{R}} \calR_{x,1}\left(m\mid w\right)\rds F_{M(1)\mid W=w}(m)\\
\nonumber & -  \int_{\mathbb{R}} \calR_{x,0}\left(m\mid w\right)\rds F_{M(0)\mid W=w}(m),\\
\label{eq:6} NDE_{x^\ast\mid w} = &\int_{\mathbb{R}} \calR_{1,x^\ast}\left(m \mid w\right)\rds F_{M(x^\ast)\mid W=w}(m)\\
\nonumber & - \int_{\mathbb{R}}\calR_{0,x^\ast}\left(m \mid w\right)\rds F_{M(x^\ast)\mid W=w}(m).
\end{align}

\subsection{Rank ordering}
\label{sec:3.2}
Let $u = F_{M(x^\ast)\mid W=w}(m)$, for $m \in \mathbb{R}$, and let $\xi^{x^\ast}_{u} = F^{-1}_{M(x^\ast)\mid W=w}(u)$ be the $u$th quantile of $M(x^\ast) \mid W=w$, for $u \in (0,1)$. Also, let $U_{M(x^\ast) \mid W=w} = F_{M(x^\ast)\mid W=w}(M(x^\ast)\mid W=w)$ be the rank transform of $M(x^\ast)$ given $W=w$, which follows a standard uniform distribution. Similarly to \eqref{eq:4}, we define the function
\begin{equation}\label{eq:7}
R_{x,m,x^\ast}(u\mid w) = \expect(Y(x,m) \mid  U_{M(x^\ast)\mid W=w} = u), \quad 0 < u < 1.
\end{equation}
The functions in \eqref{eq:4} and \eqref{eq:7} have both the same purpose. However, they are interpreted differently. Expression \eqref{eq:4} concerns the conditional expectation of $Y(x,m)$ given a fixed value $m$ of $M(x^\ast)$, which may correspond to different quantile levels (probabilities) of the mediator under treatment ($x^\ast=1$) or under control ($x^\ast=0$). In contrast, expression \eqref{eq:7} concerns the conditional expectation of $Y(x,m)$ given a fixed quantile level $u$, which may correspond to different values of the mediator under treatment  ($x^\ast=1$) or under control  ($x^\ast=0$) .

\begin{remark}
We can establish a relationship between \eqref{eq:4} and \eqref{eq:7} as a consequence of the one-to-one relationship between $u$ and $\xi^{x^\ast}_{u}$. In particular, we can assume that the function $R_{x, x^\ast}(\cdot \mid w)$ is the result of the composition $\calR_{x, x^\ast} \circ F^{-1}_{M(x^\ast)\mid W=w}$ such that $R_{x, \xi^{x^\ast}_{u}, x^\ast}(u \mid w) = \mathcal{R}_{x,x^\ast}\left(\xi^{x^\ast}_{u} \mid w\right)$ for every $u = F_{M(x^\ast) \mid W=w}(\xi^{x^\ast}_{u})$.
\end{remark}

The natural indirect effect in \eqref{eq:5} can be written as follows:
\begin{align}\label{eq:8}
NIE_{x \mid w} = & \int_{\mathbb{R}} \calR_{x,1}\left(m\mid w\right)\rds F_{M(1)\mid W=w}(m) - \int_{\mathbb{R}}\calR_{x,0}\left(m\mid w\right)\rds F_{M(0)\mid W=w}(m)\\
\nonumber = & \int_{0}^1 \calR_{x,1}\left(F^{-1}_{M(1)\mid W=w}(u)\mid w\right)\rds u - \int_{0}^1 \calR_{x,0}\left(F^{-1}_{M(0)\mid W=w}(u) \mid w\right) \rds u\\
\nonumber = & \int_{0}^1 \calR_{x,1}\left(F^{-1}_{M(1)\mid W=w}(u) \mid w\right) - \calR_{x,0}\left(F^{-1}_{M(0)\mid W=w}(u) \mid w\right) \rds u\\
\nonumber = & \int_{0}^1 R_{x,\xi^{1}_{u}, 1}\left\{F_{M(1)\mid W=w}\left(F^{-1}_{M(1)\mid W=w}(u)\right) \mid w\right\}\\
\nonumber & - R_{x,\xi^{0}_{u}, 0}\left\{F_{M(0)\mid W=w}\left(F^{-1}_{M(0)\mid W=w}(u)\right) \mid w\right\}\rds u,
\end{align}
where the second equality follows from the substitution $u= F_{M(x^\ast)\mid W=w}(m)$, $x^\ast=0,1$; the third equality follows from the property of linearity of integrals; and the fourth equality follows from the relationship
between $R_{x,\xi^{x^\ast}_{u},x^\ast}\left( u \mid w\right)$ and $\calR_{x,x^\ast}\left(\xi_u^{x^\ast} \mid w \right)$ and the identity
\begin{equation}\label{eq:9}
R_{x,\xi^{x^\ast}_{u},x^\ast}\left( u \mid w \right) = R_{x,\xi^{x^\ast}_{u},x^\ast}\left\{ F_{M(x^\ast)\mid W=w}\left(F_{M(x^\ast)\mid W=w}^{-1}\left(u\right)\right)  \mid w \right\}.
\end{equation}

Consider the last integrand of the $NIE_{x \mid w}$ in \eqref{eq:8} and define
\begin{align*}
NIE_{x\mid u, w} = & \,\, R_{x,\xi^{1}_{u},1}\left\{F_{M(1)\mid W=w}\left(F^{-1}_{M(1)\mid W=w}(u)\right) \mid w\right\}\\
& - R_{x,\xi^{0}_{u},0}\left\{F_{M(0)\mid W=w}\left(F^{-1}_{M(0)\mid W=w}(u)\right)  \mid w\right\}.
\end{align*}
This quantity can be interpreted as the natural indirect effect of $X$ on $Y$ at the quantile $u$ of $M(x^\ast) \mid W=w$, to which we refer as the \textit{$u$-specific indirect effect} of $X$ on $Y$, conditional on $W$. By using similar arguments as above, we also obtain
\begin{align}\label{eq:10}
ACE_w  = & \expect\left(Y(1) \mid W=w\right) - \expect\left(Y(0)\mid W=w\right)\\
\nonumber = & \expect\left(Y(1, M(1))\mid W=w\right) - \expect\left(Y(0, M(0))\mid W=w\right)\\
\nonumber = & \int_{\mathbb{R}} \calR_{1,1}(m\mid w) \rds F_{M(1)\mid W=w}(m\mid w) -\int_{\mathbb{R}}  \calR_{0,0}(m\mid w)  \rds F_{M(0)\mid W=w}(m)\\
\nonumber = & \int_{0}^1 R_{1, \xi_u^1,1} \left\{F_{M(1)\mid W=w}(F^{-1}_{M(1)\mid W=w}(u))\mid w\right\}\\
\nonumber & - R_{0, \xi_u^0,0}\left\{F_{M(0)\mid W=w}(F^{-1}_{M(0)\mid W=w}(u)) \mid w\right\}  \rds u.
\end{align}
We refer to the last integrand in \eqref{eq:10}
\begin{align*}
ACE_{u,w} = & \,\, R_{1, \xi_u^1,1} \left\{F_{M(1)\mid W=w}(F^{-1}_{M(1)\mid W=w}(u))\mid w\right\}\\
& - R_{0, \xi_u^0,0}\left\{F_{M(0)\mid W=w}(F^{-1}_{M(0)\mid W=w}(u)) \mid w\right\}
\end{align*}
as the \textit{$u$-specific total effect} of $X$ on $Y$, and we define the \textit{$u$-specific direct effect} as the difference between the $u$-specific total effect and $u$-specific indirect effect, that is,
\begin{align*}
NDE_{x^\ast\mid u,w} = & \,\, R_{1, \xi_u^{x^\ast},x^\ast} \left\{F_{M(x^\ast)\mid W=w}(F^{-1}_{M(x^\ast)\mid W=w}(u))\mid w\right\}\\
& - R_{0, \xi_u^{x^\ast},x^\ast}\left\{F_{M(x^\ast)\mid W=w}(F^{-1}_{M(x^\ast)\mid W=w}(u)) \mid w \right\}.
\end{align*}

\subsection{Identifiability}
\label{sec:3.3}
To identify and estimate natural direct and indirect effects, sequentially ignorability assumptions are usually invoked \cite[e.g.,][]{Pearl2001,VanderWeeleVansteelandt2009,imai_2010a}. Throughout the paper, we make the following assumption.
\begin{assumption}[Sequential Ignorability \citep{imai_2010a}]\label{assm:1}
\begin{minipage}[t]{1\textwidth}
     \begin{itemize}
		\item[(i)] Ignorability of the treatment: $\left(Y(x,m), M(x^{\ast})\right) \independent  X \mid W$, for $x,x^{\ast}=0,1$ and for all $m \in \R$.
		\item[(ii)] Ignorability of the mediator: $Y(x,m) \independent M(x^{\ast}) \mid X=x^{\ast}, W$, for $x,x^{\ast}=0,1$ and for all $m \in \R$.
	\end{itemize}
\end{minipage}
\end{assumption}

Let $X$, $M = M(X) = X\cdot M(1) + (1-X)\cdot M(0)$, and $Y = Y(X) = X\cdot Y(1) + (1-X)\cdot Y(0)$ be, respectively, the actual treatment, observed value of the mediator, and observed value of the outcome. Also, let's define the function
\begin{equation} \label{eq:11}
	\calR(x, m, w)=\expect\left(Y \mid X=x, M=m, W=w\right),
\end{equation}
which describes the average outcome in the subpopulation with covariates level $w$ exposed to treatment $x$ as a function of the mediator. For instance, in the NCHS birthweight study, $\calR(1, m, w)$ describes the mortality risk for infants in the subpopulation $W=w$ exposed to maternal smoking, as a function of birthweight. Analogously, $\calR(0,m,w)$ describes the mortality risk for unexposed infants.

Under Assumption \ref{assm:1},
$$
\calR_{x, x^\ast}(m \mid  w)=\calR(x, m, w)
$$
for $x^\ast=0,1$, and for each $x=0,1$, $m \in \R$ and $w$ (see Appendix), and the mediation formula \citep{Pearl2001} holds:
\begin{align*}
& \expect \left(Y(x, M(x^\ast)) \mid W=w\right) \\
& \quad = \int_{\mathbb{R}} \mathrm{E}\left(Y \mid X=x, M=m, W=w\right) \rds F_{M\mid X=x^\ast, W=w}(m)\\
& \quad = \int_{\mathbb{R}} \calR(x, m, w) \rds F_{M\mid X=x^\ast, W=w}(m)
\end{align*}
where $F_{M| X=x^\ast, W=w}(m)$, for $m \in \mathbb{R}$, is the cumulative distribution function of $M\mid  X=x^\ast, W=w$. Therefore, the natural indirect effect $NIE_{x \mid w}$ is identifiable from the observed data and is calculated as
\begin{align}\label{eq:12}
NIE_{x\mid w} = & \int_{\R}  \expect\left(Y \mid X=x, M=m, W=w\right) \rds F_{M\mid X=1, W=w}(m) \\
\nonumber & - \int_{\R}  \expect\left(Y \mid X=x, M=m, W=w\right) \rds F_{M\mid X=0, W=w}(m)\\
\nonumber = & \int_{\mathbb{R}} \calR(x, m, w) \rds F_{M\mid X=1, W=w}(m)- \int_{\mathbb{R}} \calR(x, m, w) \rds F_{M\mid X=0, W=w}(m).
\end{align}

Similar to \eqref{eq:11}, we now define the function:
\begin{equation}\label{eq:13}
R(x,x^\ast,u,w)=\expect\left(Y \mid X=x, U_{M\mid X=x^\ast, W=w} =u, W=w\right)
\end{equation}
where $U_{M\mid X=x^\ast, W=w}$ is the rank transform of $M\mid X=x^\ast, W=w$, which follows a standard uniform distribution.
We have that
$$
R(x,x^\ast,u,w) = \calR(x,\xi_{u\mid x^\ast},w)
$$
for every $u=F_{M\mid X=x^\ast, W=w}(\xi_{u\mid x^\ast})$, where   $\xi_{u\mid x^\ast} = F^{-1}_{M\mid X=x^\ast, W=w}(u)$ is the $u$th quantile of the random variable $M\mid X=x^\ast, W=w$  for all $u \in (0,1)$.

The natural indirect effect in \eqref{eq:12} can be then written as follows:
\begin{align}\label{eq:14}
NIE_{x|w} = &  \int_{\mathbb{R}} \calR\left(x,m,w\right) \rds F_{M\mid X=1, W=w}(m) - \int_{\mathbb{R}} \calR\left(x,m,w\right)\rds F_{M\mid X=0, W=w}(m) \\
\nonumber = & \int_{0}^1 \calR\left(x,F^{-1}_{M\mid X=1, W=w}(u),w\right)\rds u - \int_{0}^1 \calR\left(x,F^{-1}_{M\mid X=0, W=w}(u),w\right) \rds u \\
\nonumber = & \int_{0}^1 \calR\left(x,F^{-1}_{M\mid X=1, W=w}(u),w\right) -  \calR\left(x,F^{-1}_{M\mid X=0, W=w}(u),w\right) \rds u\\
\nonumber = & \int_{0}^1 R\left\{x,1,F_{M\mid X=1, W=w}\left(F^{-1}_{M\mid X=1, W=w}(u)\right), w\right\} \\
\nonumber & - R\left\{x,0,F_{M\mid X=0, W=w}\left(F^{-1}_{M\mid X=0, W=w}(u)\right),w \right\}\rds u,
\end{align}
where the second equality follows from the substitution $u= F_{M\mid X=x^\ast, W=w}(m)$, $x^\ast=0,1$; the third equality follows from the property of linearity of integrals; and the fourth equality follows from the relationship
between $R\left(x,x^\ast, u,w \right)$ and $\calR \left(x,\xi_{u \mid x^\ast},w \right)$ and the identity
\begin{equation}\label{eq:15}
R\left(x, x^\ast, u, w\right) = R \left\{x, x^\ast, F_{M \mid X=x^\ast, W=w}\left(F_{M\mid X=x^\ast, W=w}^{-1}\left(u\right)\right),w \right\}.
\end{equation}
Under Assumption \ref{assm:1}, $R_{x,\xi^{x^\ast}_u,x^\ast}(u \mid w)=R(x,x^\ast,u,w)$. Therefore, the following proposition holds:
\begin{proposition}\label{theo:1}
Under Assumption \ref{assm:1}, we have
\begin{align}\label{eq:16}
NIE_{x|u, w} = & \,\, R\left\{x,1,F_{M\mid X=1, W=w}\left(F^{-1}_{M\mid X=1,  W=w}(u)\right),w\right\}\\
\nonumber & - R\left\{x,0,F_{M\mid X=0, W=w}\left(F^{-1}_{M\mid X=0,  W=w}(u)\right), w\right\}.
\end{align}
\end{proposition}
The proof of Proposition~\ref{theo:1} is given in Appendix.

\subsection{Decomposition of the indirect effects}
\label{sec:3.4}

Given Assumption~\ref{assm:1} and expression \eqref{eq:16}, the $u$-specific natural indirect effect is equal to the derivative of \eqref{eq:15} with respect to $x^\ast$. To illustrate this fact, it may be instructive to temporarily consider $R\left(x,x^\ast,u, w\right)$ as a differentiable function of a continuous variable $x^\ast$. By the chain rule,  the first derivative of $R(\cdot,x^\ast,\cdot,\cdot)$ with respect to $x^\ast$ gives

\begin{align}\label{eq:17}
& \dfrac{\rds R\left\{x,x^\ast,F_{M\mid X=x^\ast, W=w}\left(F^{-1}_{M\mid X=x^\ast,W=w}(u)\right),w\right\}}{\rds x^\ast}\\
\nonumber & \quad = \dfrac{\rds R(x,x^\ast, u, w)}{\rds u} \cdot  \dfrac{\rds F_{M\mid X=x^\ast, W=w}\left(F^{-1}_{M\mid X=x^\ast, W=w}(u) \right)}{\rds F^{-1}_{M\mid X=x^\ast, W=w}(u)} \cdot \dfrac{\rds  F^{-1}_{M\mid X=x^\ast, W=w}(u)  }{\rds  x^\ast} \nonumber\\
\nonumber & \quad =  r(x,x^\ast, u,w) \cdot [s( u,x^\ast, w)]^{-1} \cdot q(u, x^\ast, w),
\end{align}
where
\begin{itemize}
\item $r(x,x^\ast, u, w)$ denotes the derivative of $R(x,x^\ast, u, w)$ with respect to $u$,
\item $s(u, x^\ast, w) =  \rds F^{-1}_{M\mid X=x^\ast, W=w}(u)/ \rds u$ is the derivative of the conditional quantile function of $M$ given $X=x^\ast$ and $W=w$ with respect to $u$, and
\item $q(u, x^\ast, w) = \rds F_{M\mid X=x^\ast, W=w}^{-1}(u)/\rds x^\ast$ denotes the derivative of the conditional quantile function of $M$ given $X=x^\ast$ and $W=w$ with respect to $x^\ast$.
\end{itemize}

The function $s(u,x^\ast,w)$ is known as \textit{sparsity function} \citep{tukey_1965} or quantile-density function \citep{parzen_1979}, and from the identity
$$F_{M\mid X=x^\ast, W=w}\left(F_{M\mid X=x^\ast, W=w}^{-1}(u)\right) = u,$$
it follows that
$s(u,x^\ast,w) = f^{-1}_{M \mid X=x^\ast, W=w}\left(\xi_{u\mid x^\ast} \right)$. That is, the sparsity function is the reciprocal of the density function evaluated at the quantile of interest and it is used as a measure of local variability. Variability is higher where the data are more sparse (less dense) and, vice versa, lower where the data are less sparse (more dense).

In our exposition, $X$ is binary, therefore differentiation with respect to $x^\ast$ should be loosely interpreted as differencing between adjacent levels of $x^\ast$, with $\rds x^\ast=1$ and
\begin{align}\label{eq:18}
& \dfrac{\rds R\left\{x,x^\ast,F_{M\mid X=x^\ast, W=w}\left(F^{-1}_{M\mid X=x^\ast, W=w}(u)\right), w\right\}}{\rds x^\ast}\\
\nonumber & \quad = R\left\{x,1,F_{M\mid X=1, W=w}\left(F^{-1}_{M\mid X=1, W=w}(u)\right),w\right\}\\
\nonumber & \quad\quad - R\left\{x,0,F_{M\mid X=0, W=w}\left(F^{-1}_{M\mid X=0, W=w}(u)\right),w \right\}\\
\nonumber & \quad = NIE_{x|u, w}.
\end{align}
Note that, in the binary case, we have that $q(u, x^\ast, w) =  F^{-1}_{M\mid X=1, W=w}(u) - F^{-1}_{M\mid X=0, W=w}(u)$. Moreover, under Assumption~\ref{assm:1} (ignorability of the treatment), we have that
\begin{equation}\label{eq:19}
F^{-1}_{M(1)\mid W=w}(u) - F^{-1}_{M(0)\mid W=w}(u)  = F^{-1}_{M\mid X=1, W=w}(u) - F^{-1}_{M\mid X=0, W=w}(u).
\end{equation}
Therefore $q(u, x^\ast, w)$ is the $u$th quantile effect of $X$ on $M$ .

The connection between \eqref{eq:17} and \eqref{eq:18} is given by the following proposition:
\begin{proposition}\label{theo:2}
Let $R(\cdot,x^\ast,\cdot,\cdot) \in C^{1}$ be a continuously differentiable function of $x^\ast$, where $x^\ast = (\zeta^\ast - a)/(b - a)$ and $\zeta^\ast \in [a,b] \subseteq \R$. Then, there exists some $\tilde{x}^\ast \in (0,1)$ such that
\begin{align}\label{eq:20}
& \dfrac{\rds R\left\{x,x^\ast,F_{M\mid X=x^\ast, W=w}\left(F^{-1}_{M\mid X=x^\ast,W=w}(u)\right),w\right\}}{\rds x^\ast}\Bigg\rvert_{x^\ast = \tilde{x}^\ast}\\
\nonumber & \quad =  r(x,x^\ast, u,w)\bigg\rvert_{x^\ast = \tilde{x}^\ast} \cdot [s(u,x^\ast, w)]^{-1}\bigg\rvert_{x^\ast = \tilde{x}^\ast} \cdot q(u, x^\ast, w)\bigg\rvert_{x^\ast = \tilde{x}^\ast}\\
\nonumber & \quad = R(x,1,u,w) - R(x,0,u,w).
\end{align}
\end{proposition}
The proof of Proposition~\ref{theo:2} is given in Appendix.

\begin{remark}
In general, the value $\tilde{x}^\ast$ is not unique. However, if in addition $\rds R(\cdot,x^\ast,\cdot,\cdot)/\rds x^\ast$ is strictly monotonic, then $\tilde{x}^\ast$ is unique.
\end{remark}

The above proposition is based on a simplification of how a binary exposure is defined. We can interpret $\zeta^\ast$ as a latent exposure that gives rise to $x^\ast \in \{0,1\}$ according to a threshold mechanism. This may or may not be appropriate for some exposures. For example, in the NCHS birthweight study, $\zeta^\ast$ could be interpreted as the amount of nicotine or other toxic substances to which the fetus is exposed during pregnancy. The dichotomized $x^{\ast}$ may be considered as arising from a threshold mechanism of the type $I(\zeta^{\ast} > c)$, $c \in (0,\zeta^{\ast}_{\max})$. Note that the value $\tilde{x}^\ast$ may depend on $x$, $u$, and $w$. For simplicity, in our application (Section~\ref{sec:5}) we do not concern ourselves with the calculation of such values but, instead, introduce `average' approximations as explained in Section~\ref{sec:4}.

Under Assumption~\ref{assm:1}, Proposition \ref{theo:2} implies that we can decompose $NIE_{x\mid u}$, i.e. the $u$-specific natural indirect effect of $X$ on $Y$ at the quantile $u$ of $M$, into three components:
\begin{enumerate}
	\item $q(u, x^\ast, w)\bigg\rvert_{x^\ast = \tilde{x}^\ast} \equiv q(u, \tilde{x}^\ast, w)$, the $u$-quantile effect of the exposure on the mediator; 
	\item $[s(u,x^\ast, w)]^{-1}\bigg\rvert_{x^\ast = \tilde{x}^\ast} \equiv [s(u,\tilde{x}^\ast, w)]^{-1}$, the conditional density of the mediator given $X= \tilde{x}^\ast$ and $W=w$ at the quantile $u$;
	\item $r(x,x^\ast, u,w)\bigg\rvert_{x^\ast = \tilde{x}^\ast} \equiv r(x, \tilde{x}^\ast, u, w)$, the sensitivity of the conditional expected value of the outcome given $X=x, W=w$ to changes in the conditional distribution of the mediator given $X=\tilde{x}^\ast, W=w$ at $u$.
\end{enumerate}
Naturally, the stronger the $u$-quantile effect, the larger will be the indirect effect if $R(x, x^\ast, u,w)$ is sensitive near $u$. The density  $[s(u,\tilde{x}^\ast, w)]^{-1}$ amplifies (attenuates) the quantile effect in regions where the variability of the mediator is lower (higher).

\section{Models and inference}
\label{sec:4}

We now introduce the models and related inferential aspects. For our data analysis, we favour semi- and non-parametric approaches as we want some modelling flexibility. However, parametric alternatives, of which we mention a few, can be considered as well.

\subsection{Modelling the mediator}
\label{sec:4.1}
The first step in our modelling approach involves the conditional quantile function of the mediator. Suppose that the exposure $X$ and the $u$th quantile of $M$, adjusted for confounders $W$, are related according to a linear model of the type
\begin{equation}\label{eq:21}
F_{M\mid X=x,W=w}^{-1}(u) = \beta_{0}(u) + x\beta_{1}(u) + w'\gamma(u),
\end{equation}
where $\gamma(u)$ is a $p \times 1$ vector of coefficients and $w$ may contain interaction terms with $x$. It also follows that the marginal effect associated with $x$ is $q(u, x, w) = \beta_{1}(u) + \rds w'\gamma(u)/\rds x$. Note that the linear specification of the model implies that $q(u, x, w) = q(u, w)$ does not depend on $x$. If we assume that $x$ is the result of a dichotomization of a latent continuous exposure, this may represent a simplification of the true dose-response relationship between the latent exposure and the mediator. In practical situations, information on such a relationship may be unavailable or unreliable or costly to obtain. We further elaborate on this point with regard to the relationship between smoking and birthweight (Section~\ref{sec:5}).

Consider for a moment the case $\gamma(u) = 0$. Under ignorability of the treatment, since $X$ is binary, $\beta_{0}(u)$ is the $u$th quantile of $M$ in the unexposed population and $\beta_{1}(u)$ is the `quantile treatment effect' \citep{doksum_1974,lehmann_1975,koenker_xiao} on the mediator. In this case, it is straightforward to estimate these parameters using the sample quantiles $\hat{\beta}_{0}(u) = \hat{F}^{-1}_{M\mid X=0}\left(u\right)$ and $\hat{\beta}_{1}(u) = \hat{F}^{-1}_{M\mid X=1}\left(u\right) - \hat{F}^{-1}_{M\mid X=0}\left(u\right)$, $0<u<1$. For more general problems where $\gamma(u) \neq 0$, estimation can be based on simplex or interior point methods \citep{koenker_2005}. If the assumption of linearity of the quantile function does not hold, one can consider nonlinear quantile regression models \citep{koenker_park} or exploits the equivariance property of quantiles by applying a suitable transformation towards linearity \citep{geraci_jones_2015}. In any of the above cases, parametric assumptions on the functional form of $F$ are avoided in favour of weaker conditional quantile restrictions \citep{powell_1994}.

Parametric specifications of $F$ can also be of interest. An approach based on a mixture of normals is proposed by \cite{dominici_etal}. For instance, suppose that $M|(X=x,W=w) \sim \mathcal{N}(\beta_{0} + x\beta_{1} + w'\gamma, \sigma^{2})$, then $F_{M|X=x,W=w}^{-1}(u) = \beta_{0}(u) + x\beta_{1} + w'\gamma$, where $\beta_{0}(u) = \beta_{0} + \sigma\Phi^{-1}(u)$ and $\Phi$ denotes the standard normal distribution function. It follows that $q(u, x, w) = \beta_{1} + \rds w'\gamma/\rds x$ is constant with respect to $u$. That is, the quantile regression curves are simply vertical translations of one another. Moreover, under normal assumptions, the decomposition of the $u$-specific indirect effect given in \eqref{eq:20} would simplify to
\begin{equation*}
\rds R(x, \tilde{x}^\ast, u, w) = r(x,\tilde{x}^\ast, u,w)\cdot \frac{\beta_{1} + \rds w'\gamma/\rds x\big\rvert_{x^\ast = \tilde{x}^\ast}}{\sqrt{2 \pi}\sigma}\cdot \exp\left[-\frac{\left\{\Phi^{-1}\left(u\right)\right\}^{2}}{2}\right],
\end{equation*}
which could be regarded as a null hypothesis model. However, location-shift or even location--scale-shift effects only may fail to capture the complexity of the distributional relationships between variables.

\subsection{Modelling the outcome}
\label{sec:4.2}
Let us now consider the modelling of the function $R(x, x^\ast, u, w)$, $0<u<1$. For ease of exposition, we refer to the NCHS data analysis, where the focus is on mortality. We can estimate $R(x, x^\ast, u, w)$ by a sequence of mortality rates at different quantiles of a sample $\left\{m_{1},\ldots,m_{n}\right\}$ of observations of $M$, where $n$ denotes the sample size. Let $(u_{k-1}, u_{k}]$, $k = 1, \ldots, K$, be a sequence of bins that partition the unit interval, with $u_{0} = 0$ and $u_{K} = 1$ (the leftmost bin is treated as closed). Next, we classify the observations according to the distribution of $M\mid X=x^\ast,W=w$. If $\hat{F}_{M\mid X=x^\ast, W=w}(m_{i})$ falls in $(u_{k-1}, u_{k}]$, then we assign the $i$th observation to the $k$th bin. The latter task can be easily achieved by noting that $F_{M\mid X=x^\ast,W=w}$ is simply the inverse of the conditional quantile function. Using the midpoints $\bar{u}_{k}$, $k = 1, \ldots, K$, as representative of each class $(u_{k-1}, u_{k}]$, we can	
\begin{enumerate}
	\item estimate $\hat{F}_{M\mid X=x^\ast,W=w}^{-1}(\bar{u}_{k})$ in \eqref{eq:21};
	\item calculate the predicted values $\hat{F}_{i\mid X=x^\ast,W=w}^{-1}(\bar{u}_{k})$ and the absolute residuals $e_{ik} = \left|m_{i} - \hat{F}_{i\mid X=x^\ast,W=w}^{-1}(\bar{u}_{k})\right|$, $i = 1,\ldots,n$, $k = 1,\ldots,K$;
	\item for a given observation $i$, find the value $k$, $k = 1,\ldots,K$, such that $e_{ik}$ is smallest;
	\item assign the observation $i$ to the $k$th bin for the value $k$ determined in the previous step.
\end{enumerate}

Finally, let $n^{(k)}_{x}$ and $z^{(k)}_{x}$ be, respectively, the population at risk and the number of deaths in the $k$th bin by exposure status, with $\sum_{k=1}^{K}n^{(k)}_{x} = n_{x}$. In practice, the choice of $K$ might depend on obtaining a reasonable number of events in the bins. The rate $\hat{R}^{(k)}(x, x^\ast, u, w) = z^{(k)}_{x}/n^{(k)}_{x}$ is an estimate of the mortality risk for infants falling in the birthweight quantile class $(u_{k-1}, u_{k}]$ who were either unexposed ($x = 0$) or exposed ($x = 1$) to smoking.

As we have seen in \eqref{eq:20}, the derivative $r(x, \tilde{x}^\ast, u, w)$ is a component of the $u$-specific indirect effect in \eqref{eq:9} and, as we will see in the next section, plays an important role. If the risk is constant over $u$, then $r(x, \tilde{x}^\ast, u, w)  = 0$ and $\rds R(x, \tilde{x}^\ast, u, w)/\rds x^\ast = 0$ for any value of the other two components of the indirect effect. In our data analysis, we estimate $r(x, x^\ast, u, w)$ numerically with
\begin{equation}\label{eq:22}
\hat{r}(x, x^\ast, \tilde{u}, w) = \frac{\tilde{R}\left(x, x^\ast, \tilde{u} + \delta_{n}, w\right)-\tilde{R}\left(x, x^\ast,\tilde{u} - \delta_{n},w\right)}{2\delta_{n}},
\end{equation}
for some $\tilde{u}$, $\bar{u}_{1} < \tilde{u} < \bar{u}_{K}$, where $\tilde{R}\left(x, x^\ast, u, w\right)$ is an interpolation function (e.g., linear or spline) of the points $\left(\bar{u}_{k},\hat{R}^{(k)}(x, x^\ast, u, w)\right)$ and $\delta_{n} \xrightarrow{n \to \infty} 0$ is a suitably small bandwidth parameter (for example, see \eqref{eq:24} below). In a nonparametric approach, the calculation for $x^\ast = \tilde{x}^{\ast}$ can be obtained by taking a weighted average of $\hat{r}(x, 0, \tilde{u}, w)$ and $\hat{r}(x, 1, \tilde{u}, w)$, with weights proportional to the sample sizes of the two groups $n_{x}$, $x = 0,1$.

\sloppy Again, one can follow a parametric approach to the modelling of $R(x, x^\ast,u,w)$, with the added benefit of possibly obtaining  $r(x, x^\ast, u, w)$ analytically, in which case $r(x, \tilde{x}^\ast, u, w)$ can be obtained by replacing for $x^\ast = \tilde{x}^\ast$. Some alternatives are given by generalised linear models \citep{mccullagh_nelder}, additive models \citep{hastie} and additive models for location, scale, and shape \citep{rigby_2005}.

\subsection{Sparsity}
\label{sec:4.3}

We now consider the nonparametric estimation of the sparsity function $s(u,\tilde{x}^\ast, w)$, which concludes the discussion on estimation. A simple approach consists in using the difference quotient \citep{koenker_2005}
\begin{equation}\label{eq:23}
\hat{s}(u,x^\ast, w) = \dfrac{\hat{F}^{-1}_{M \mid X=x^\ast, W=w}\left(u + \epsilon_{n}\right) - \hat{F}^{-1}_{M \mid X=x^\ast, W=w}\left(u - \epsilon_{n}\right)}{2\epsilon_{n}},
\end{equation}
where $\epsilon_{n} \xrightarrow{n \to \infty} 0$ is the bandwidth parameter. To this end, we consider \citeauthor{bofinger_1975}'s (\citeyear{bofinger_1975}) bandwidth
\begin{equation}\label{eq:24}
\epsilon_{n} = n^{-1/5} \left[4.5 \left\{\phi\left(v\right)\right\}^{4}/(2v^{2} + 1)^2\right]^{1/5},
\end{equation}
where $v = \Phi^{-1}\left(u\right)$ and $\phi \equiv \Phi'$. Numerical adjustments need to be introduced as appropriate in those instances where $\hat{s}(u,x^\ast, w) \leq 0$. As before, the calculation for $x^\ast = \tilde{x}^{\ast}$ can be obtained from the weighted average of $\hat{s}(u,0, w)$ and $\hat{s}(u,1, w)$, with weights proportional to $n_{x}$, $x = 0,1$.

\subsection{Standard errors}
\label{sec:4.4}
In general, inference on the indirect effects and related components may be challenging. Bootstrap \citep{efron_tibshirani} represents a flexible method to derive standard errors and perform inference. Given the large size of the NCHS dataset (more than 11 million observations), we implemented the method by \cite{kleiner_2014}. The general idea is to sample without replacement $S$ subsets of size $b$ from the original dataset of size $n$, with $b < n$. Bootstrapping of the statistic of interest is then performed on each subset and the results are averaged across $S$ subsets. This strategy, called `bag of little bootstraps', greatly reduces the computing cost when $n$ is large \citep[see][for more details]{kleiner_2012,kleiner_2014}. The confidence intervals reported in the NCHS data analysis were obtained using $b=50$ replications.

\clearpage

\section{United States infant mortality}
\label{sec:5}

\subsection{Motivation for using $u$-specific effects}
\label{sec:5.1}
In this section we briefly discuss the functions $\calR_{x, x^\ast}(m \mid w)$ \eqref{eq:4} and $R_{x, \xi^{x^\ast}_u, x^\ast}(u \mid w)$ \eqref{eq:7} in more detail to clarify the advantages of using a quantile-based approach to mediation analysis and to motivate the application of such an approach to the NCHS birthweight data.

As shown in \eqref{eq:8}, the natural indirect effect can be obtained using either $\calR_{x, x^\ast}$ or $R_{x, \xi^{x^\ast}_u, x^\ast}$. In the former case, the integral in \eqref{eq:8} is a Riemann-Stieltjes integral; in the latter, it is a Lebesgue integral. Thus the only difference is whether or not observations are ranked before summation. It is argued that `[q]uantile thinking defines statistics as summation done by sorting (ranking) data before adding' \citep[][p.654]{parzen_2004}. As a result, the contrast between, say, $R_{1,\xi_u^{1}, 1}\left(u \mid w\right)$ and $R_{0,\xi_u^{0}, 0}\left(u \mid w\right)$ is done on the same footing since $F_{M(1)}(\xi_u^{1}) = F_{M(0)}(\xi_u^{0})$.

\begin{figure}
	\includegraphics[width = \textwidth]{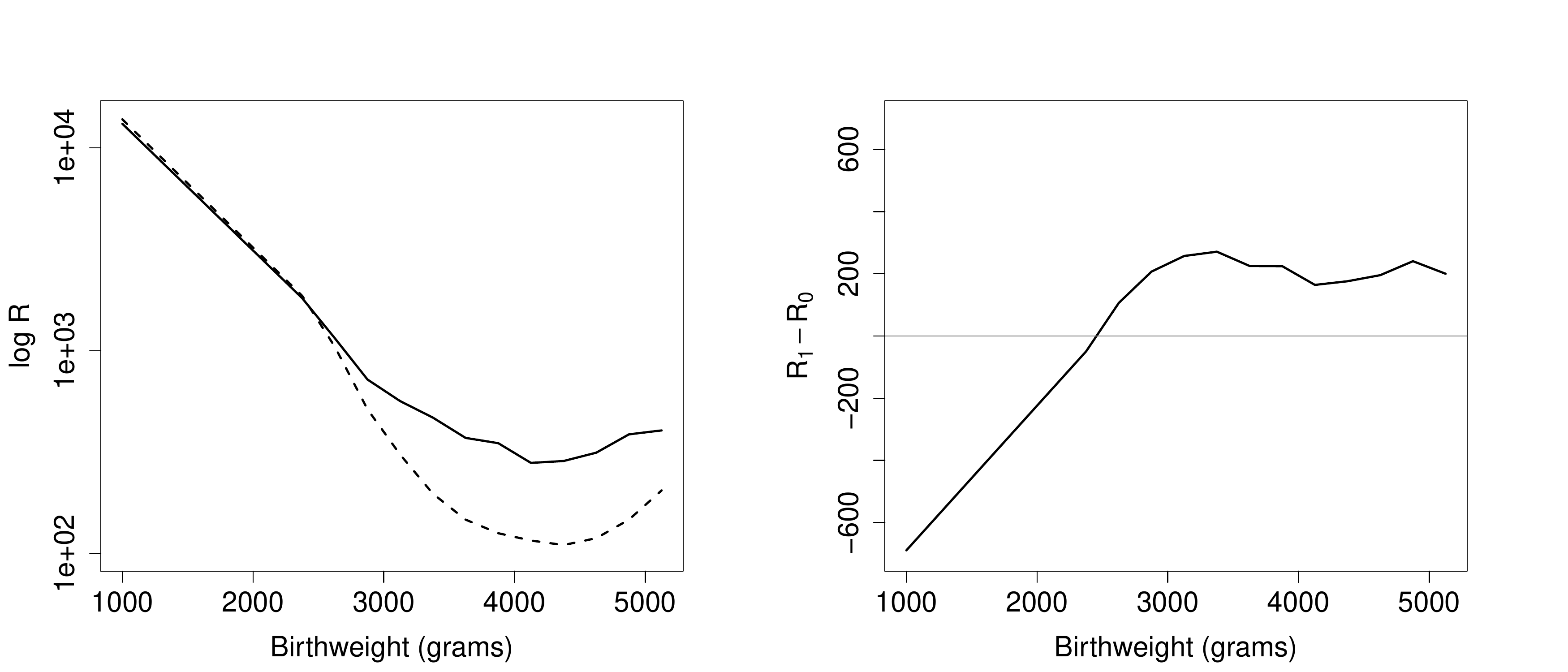}
	\caption{All-cause infant mortality (per $100,\!000$ livebirth-years) by birthweight in white singletons, National Center for Health Statistics data, United States 2001-2005. Left plot: estimates of $\calR_{0,0,m}$ (dashed line) and $\calR_{1,1,m}$ (solid line). Right plot: estimate of the difference $\calR_{1,1,m} - \calR_{0,0,m}$.}
	\label{fig:1}
\end{figure}

To exemplify, let's consider the relationship between maternal smoking, birthweight, and infant mortality as discussed in Section~\ref{sec:2}, and momentarily ignore pretreatment covariates. Under Assumption~\ref{assm:1} we have $F_{M(x^\ast)}=F_{M\mid X=x^\ast}$, $\calR_{x,x^\ast}\left(m\right) =\calR\left(x,m\right)$, for $x^\ast=0,1$, and $R_{x,\xi_u^{x^\ast}, x^\ast}\left(u \right)=R\left(x,x^\ast, u \right)$. The functions $\calR\left(1,m \right)$ and $\calR\left(0,m \right)$ describe the expected mortality risk of infants in the unexposed and exposed subgroups, respectively, as a function of birthweight. Figure~\ref{fig:1} shows their estimates per $100,\!000$ livebirth-years (phby)  on the logarithmic scale, along with their difference, for birthweights in the range $1000$ to $5250$ grams. As we can see in Figure~\ref{fig:1}, the curves $\calR\left(0,m \right)$ and $\calR\left(1,m \right)$ cross. In the past the cross-over of the curves, also referred to as `birthweight paradox', has been used to question the harmful effects of maternal smoking \citep{yerushalmy}. We argue that the birthweight paradox is related to the fact that, in general, $F_{M(0)}(m) \neq F_{M(1)}(m)$.  For instance, for the NCHS data we obtain $\hat{F}_{M\mid X=0}(2500) = 0.05$ and $\hat{F}_{M\mid X=1}(2500) = 0.10$. That is, the proportion of LBW infants born to mothers who smoked during pregnancy is twice the proportion of LBW infants born to mothers who did not smoke.

Figure~\ref{fig:2} shows the estimated mortality rates per $100,\!000$ livebirth-years as function of the quantile of $M \mid X=x^\ast$, $R\left(x,x^{\ast}, u\right)$, $x=x^\ast=0,1$, along with their difference. The birthweight paradox disappears using the quantile-based approach.

We can appreciate this fact when using Bayes' rule. Recall that for two events $A$ and $B$, $P\left(A\left|B\right.\right) = P\left(A \cap B\right)/P(B)$. If we consider the probabilities of the event $A\mid B$ under treatment and under control, we obtain 
\[
\frac{P_{1}\left(A\left|B\right.\right)}{P_{0}\left(A\left|B\right.\right)} = \frac{P_{1}\left(A \cap B\right)}{P_{0}\left(A \cap B\right)} \cdot \frac{P_{0}\left(B\right)}{P_{1}\left(B\right)},
\]
where $P_{x}(\cdot)$ expresses the probability of an event under treatment $x = 0,1$.
For instance, let $A$ be the event `death before one year of age' and $B$ the event `birthweight less than $m$ grams'. In all documented situations where the birthweight paradox arises, the birthweight distribution in the exposed population ($x = 1$) (e.g., infants born to smokers) is stochastically dominated by that in the unexposed population ($x = 0$) (e.g., infants born to nonsmokers) so that the ratio $P_{0}\left(B\right)/P_{1}\left(B\right) < 1$ throughout the distribution in the exposed population. For some values of $m$ (e.g., $m = 2500$ grams), the downplaying impact of this ratio prevails over $P_{1}\left(A \cap B\right)/P_{0}\left(A \cap B\right)$ so that the ratio of the conditional probabilities is smaller than 1. Now, the rank-ordered curves compare
\[
\frac{P_{1}\left(A\left|B\right.\right)}{P_{0}\left(A\left|B'\right.\right)} = \frac{P_{1}\left(A \cap B\right)}{P_{0}\left(A \cap B'\right)},
\]
with  $B'$ defined as `birthweight less than $m'$ grams', $m' \geq m$, so that $P_{0}\left(B'\right) = P_{1}\left(B\right)$.

\begin{figure}
	\includegraphics[width = \textwidth]{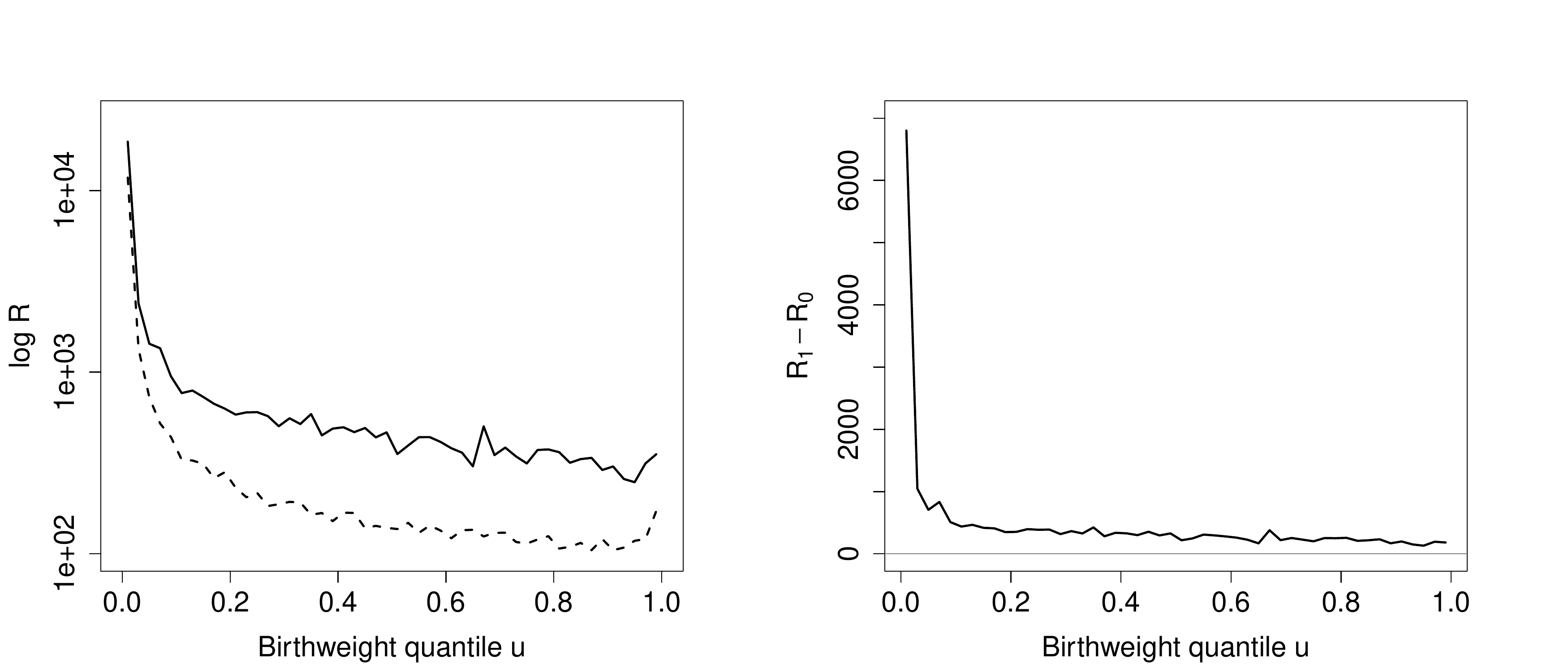}
	\caption{All-cause infant mortality (per $100,\!000$ livebirth-years) by birthweight quantile in white singletons, National Center for Health Statistics data, United States 2001-2005. Left: estimates of $R_{0,\xi_u^0,0}(u)$ (dashed line) and $R_{1,\xi_u^1,1}(u)$ (solid line). Right plot: estimates of the
	$u$-specific total effects $ACE_u=R_{1,\xi_u^1,1}(u) - R_{0,\xi_u^0,0}(u)$.}
	\label{fig:2}
\end{figure}

\subsection{All-cause mortality}
\label{sec:5.2}

In this section, we investigate infant mortality in the US population using the NCHS data introduced previously. We used birthweight as mediator ($M$) and maternal smoking ($X$) as exposure in the model $F^{-1}_{M\mid X=x}(u) = \beta_{0}(u) + x\beta_{1}(u)$. In the first of two analyses, we did not adjust the models for pretreatment variables $W$ for two reasons. First of all, we wanted to analyse the `birthweight paradox' as described in the literature \citep{wilcox_2001,hernandez_2006} to show that interesting findings may already emerge from our proposed decomposition in a preliminary analysis. Secondly, we argue that the analysis of all-cause mortality prevents any meaningful search of potential confounders given the heterogeneity of the outcome. Since it would be challenging to include in one model all the relevant factors, provided that they are known and available, in Section \ref{sec:4.2} we restrict the analysis to a specific cause of death, sudden infant death syndrome (SIDS), which allows us to limit the dimensionality of the model using as many known risk factors as possible that are available in the NCHS data.

After excluding cases with missing information on birthweight and maternal smoking, there were $11,\!590,\!581$ livebirths and $57,\!245$ deaths. Preliminarily, we examined the distributions of birthweight and gestational age in exposed and unexposed infants. While the former was shifted to the left as a consequence of the negative association between smoking and birthweight, the distribution of gestational age was similar in the two populations (results not shown). We do not discuss gestational age further in our analyses.

\begin{figure}[h!]
	\centering
	\includegraphics[scale= 0.5]{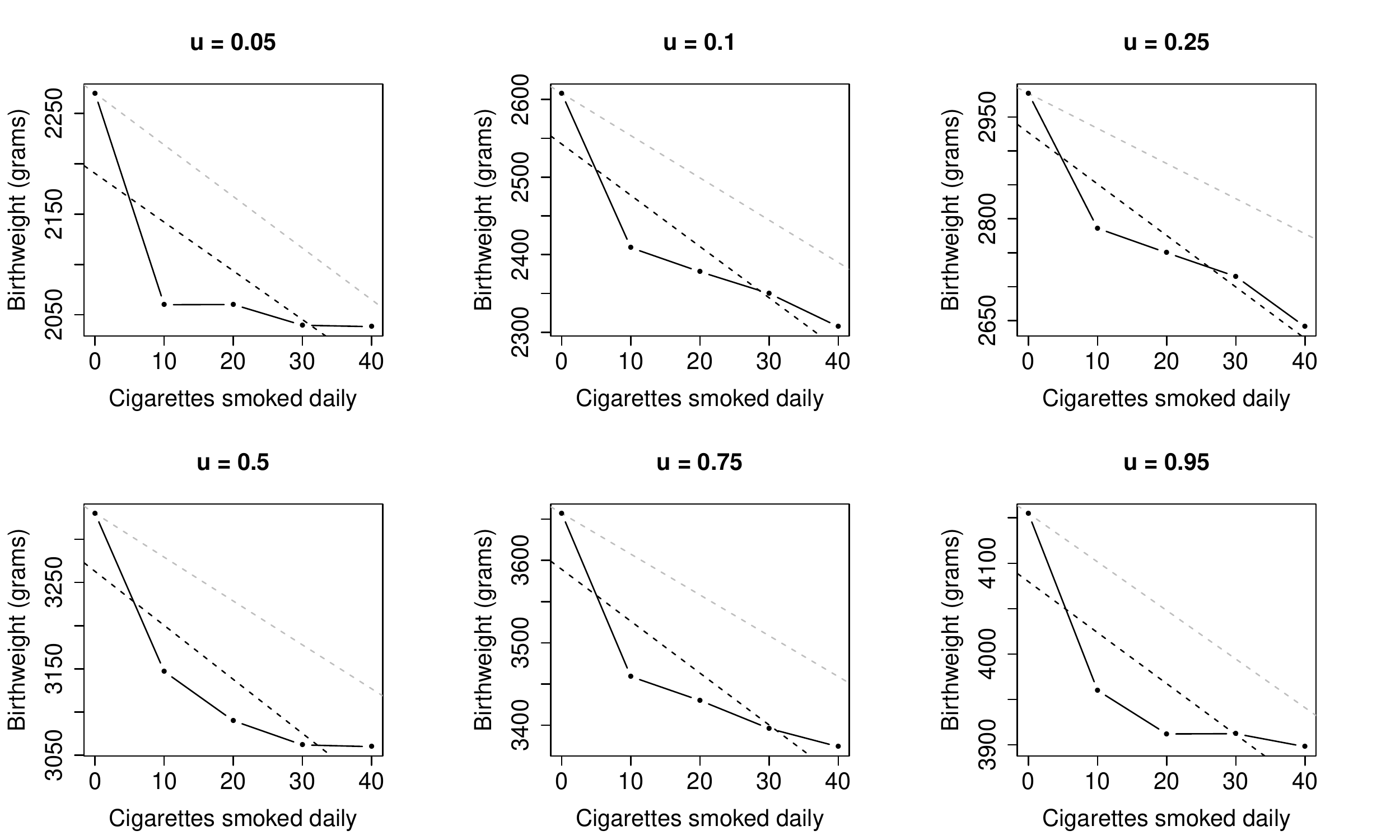}
	\caption{Birthweight quantiles and number of cigarettes smoked daily by mothers of white singletons, National Center for Health Statistics data, United States 2001-2005. Dots represent birthweight sample quantiles estimated at different levels of the exposure: 0, (0, 10], (10, 20], (20, 30], and (30, 40]. The black dashed line is an estimate of the linear dose-response relationship for a continuous exposure. The grey dashed line is an estimate of the linear relationship for a binary exposure.}
	\label{fig:3}
\end{figure}

We defined $K = 50$ intervals $(u_{k-1}, u_{k}]$ of width $0.02$, with $u_{0} = 0$ and $u_{50} = 1$. Death counts $z_{x}^{(k)}$ ranged from 52 to $23,\!999$ (median 333), totalling to $44,\!789$, in the unexposed group ($x = 0$) and from 28 to $5,\!072$ (median 134), totalling to $12,\!456$, in the exposed group ($x = 1$). The first quantile interval $(0, 0.02]$ alone accounted for, respectively, $54\%$ and $41\%$ of the overall number of deaths in the former and the latter group. The number of livebirths was $10,\!190,\!155$ in the unexposed group and $1,\!400,\!426$ in the exposed group. The estimated infant mortality rates are plotted in Figure \ref{fig:2}. Mortality in infants born to mothers who smoked during pregnancy is higher than mortality in the unexposed population across all birthweight quantiles, the former exceeding the latter by several thousands livebirth-years at lower quantiles. The risk difference becomes very small at higher birthweight quantiles.

We also investigated a dose-response relationship between birthweight and number of cigarettes smoked per day. The latter may be considered as a proxy of a latent exposure $\zeta^{\ast}$ (Proposition~\ref{theo:2}) measuring harmful substances that affect fetus growth as well as mortality risk. There was a negative gradient at all considered quantiles of birthweight (Figure~\ref{fig:3}). As compared to unexposed infants, the rate of decrease in birthweight was fastest for those exposed to up to 10 cigarettes smoked daily, and then approximately linear at higher levels of the exposure. The slope for the dichotomised exposure $X = I(\mbox{number of cigarettes} > 0)$ underestimates the change in birthweight at lower levels of the exposure, but it reasonably approximates the slope of an overall linear model for the dose-response relationship between smoking and birthweight.

The estimated $u$-specific indirect, direct, and total effects are shown in Figure~\ref{fig:4}. The indirect effect is highest at the lowest birthweight quantiles and decreases sharply with increasing $u$. More specifically, the $u$-specific indirect effect completely explains the total effect at lower $u$ but it plunges to null values near the right end of the birthweight distribution. In contrast, the $u$-specific direct effect plays a lesser role at lower quantiles than it does at higher quantiles, where it becomes the principal effect. The indirect, direct, and total effects, averaged over $u$, were equal to $243.4$, $209.9$ and $453.3$, respectively. In this case the mean effects alone do not make justice to the complexity of the picture shown in Figure~\ref{fig:4}.

The estimates of the components of the $u$-specific indirect effect as defined in \eqref{eq:20} are plotted in Figure~\ref{fig:5}. Note that the $95\%$ confidence intervals are generally narrow. The very high mortality risk differentials at lower quantiles follow from the combined effect of highly negative values of $\hat{r}(u)$ and the negative values of $\hat{\beta}_{1}(u)$, the birthweight quantile effect associated with smoking. The latter was close to the mean effect at minus $232$ grams for birthweights above the $8$th centile but, for lower quantiles, it showed larger magnitudes. However, the impact of these larger magnitudes on the estimate of the indirect effect was ultimately diminished by the higher sparsity at lower quantiles of birthweight.

\begin{figure}
	\centering
	\includegraphics[scale= 0.3]{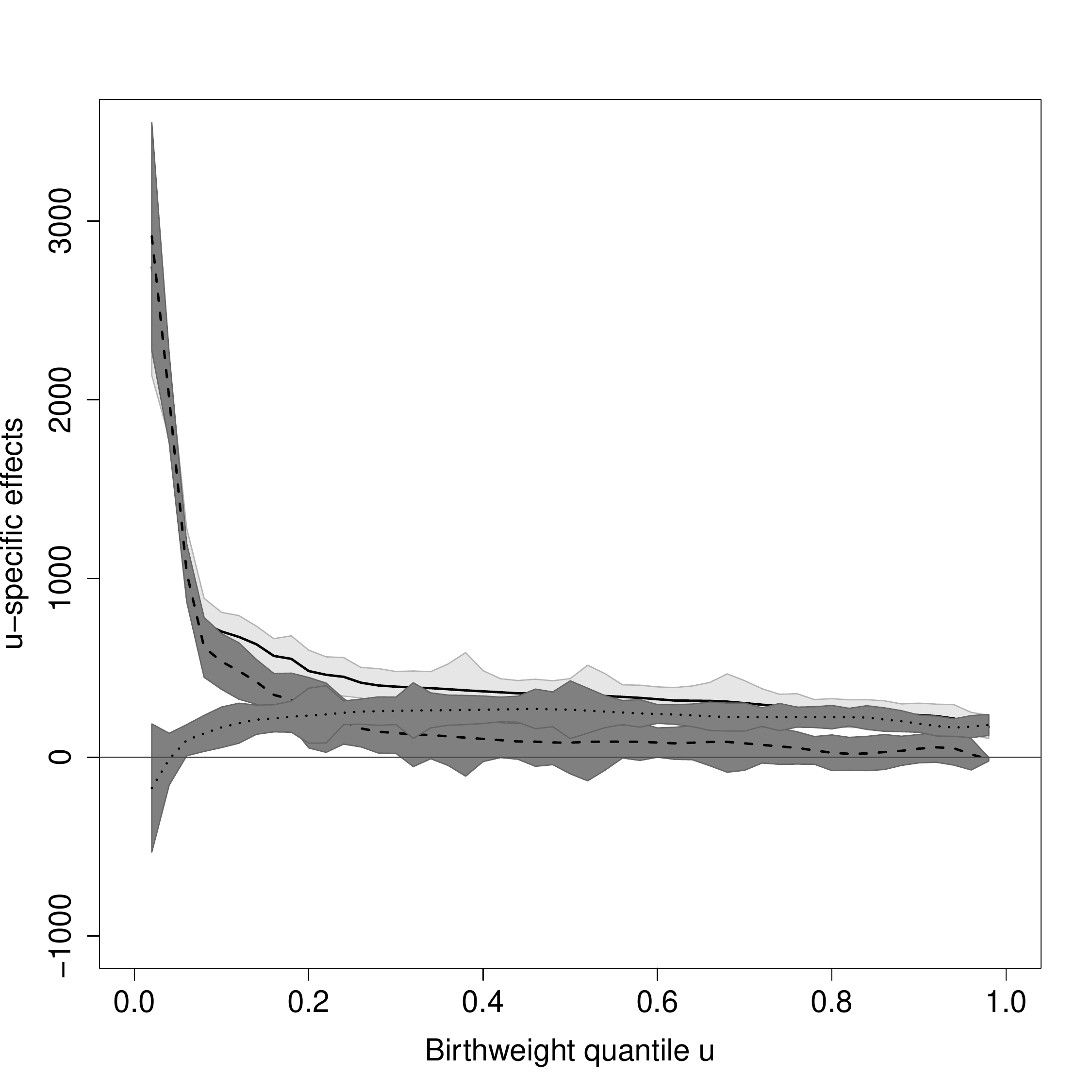}
	\caption{All-cause infant mortality (per $100,\!000$ livebirth-years) by birthweight quantile in white singletons, National Center for Health Statistics data, United States 2001-2005. Estimates of the $u$-specific indirect effect $NIE_{x = 1|u}$ (dashed line), direct effect $NDE_{x^{\ast} = 0|u}$ (dotted line), and total effect $ACE_{u}$ (solid line) of maternal smoking on infant mortality. Shaded grey areas depict $95\%$ pointwise confidence bands.}
	\label{fig:4}
\end{figure}

It has been suggested that the impact of smoking on mortality is independent of its effect on birthweight \citep{wilcox_2001}. Our results show that this is not the case, thus confirming the intuition by \cite{picciotto_2001}. This is explained by the fact that the effect of smoking on birthweight is not uniform across birthweight quantiles, nor is the effect of other well-known birthweight determinants \citep{geraci_2016}.

\begin{figure}
	\centering
	\includegraphics[width = \textwidth]{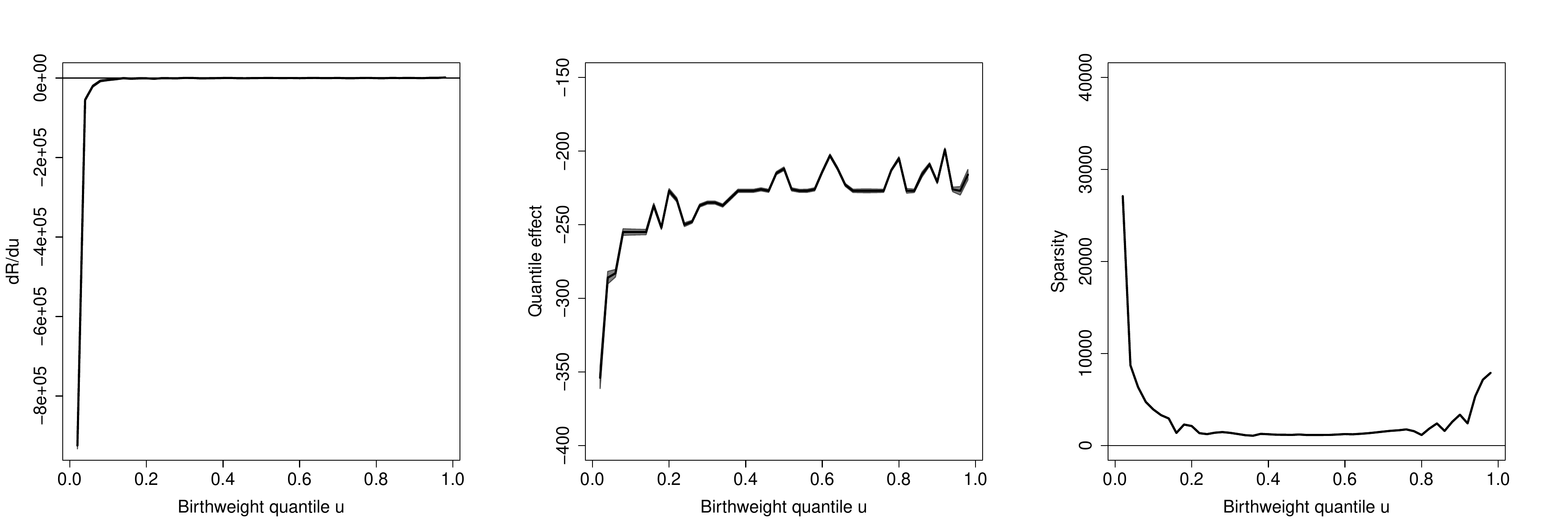}
	\caption{All-cause infant mortality (per $100,\!000$ livebirth-years) by birthweight quantile in white singletons, National Center for Health Statistics data, United States 2001-2005. Components of the estimated $u$-specific indirect effect. Shaded grey areas depict $95\%$ pointwise confidence bands.}
	\label{fig:5}
\end{figure}

As for the role of birthweight in relation to risk factors of infant mortality, our results point clearly towards the conclusion that indirect and direct effects need to be assessed at different quantiles of the birthweight distribution. The shift in birthweights plays a less important role at higher quantiles, which suggests that the effect of smoking may be acting through pathways not associated with birthweight.

It should be stressed that the interpretation of these results relies on the assumption that sequential ignorability (Assumption \ref{assm:1}) holds without conditioning on the covariates. However, it is likely that infants in the exposed and unexposed groups differ systematically in terms of maternal characteristics such as income, education, prenatal care received, alcohol consumption, as well as other factors associated with the incidence of smoking during pregnancy; all factors which are also associated with birthweight and infant mortality. As explained at the beginning of this section, we therefore restricted the analysis to SIDS mortality.

\subsection{Sudden infant death syndrome mortality}
\label{sec:5.3}

SIDS is defined as sudden death of an infant aged less than one year that remains unexplained after a thorough case investigation that includes an autopsy, a death scene investigation, and a review of the clinical history of the parents and the infant \citep{willinger_1991, shah_2006}. We calculated overall mortality rates by cause of death using the NCHS data. Between 2001 and 2005, SIDS was the first cause of mortality in US white singleton infants with a rate equal to $48.2$ phby, followed by extremely low birthweight or extreme immaturity ($45.1$ phby), and congenital heart malformations ($33.0$ phby). SIDS rate for infants born to smokers was $5$ times ($95\%$ confidence interval: $4.7$ to $5.3$) the rate for those born to nonsmokers.

As in the previous section, we modelled the quantiles of birthweight as a function of smoking. In addition, we controlled for additional SIDS risk factors ($W$) as reported in the literature \citep{schlaud_1996,leach_1999}. Under Assumption \ref{assm:1}, we considered the following linear regression model
\begin{align}\label{eq:25}
F^{-1}_{M\mid X=x, W=w}(u) = & \,\, \beta_{0}(u) + x\beta_{1}(u) + w_{1}\gamma_{1}(u) + w_{2}\gamma_{2}(u)\\
\nonumber & + w_{3}\gamma_{3}(u) + xw_{3}\gamma_{4}(u),
\end{align}
where $w_{1}$ is a dummy variable for whether prenatal care was received at any stage of the pregnancy (baseline: women who received prenatal care), $w_{2}$ is a dummy variable for maternal age at delivery less than 20 years, regardless of the total birth order (baseline: women aged 20 years or older), and $w_{3}$ is a dummy variable for alcohol consumption during pregnancy (baseline: women who did not consume alcohol). The parameter $\gamma_{4}(u)$ is associated with the interaction between smoking and alcohol. The quantile effect associated with smoking is then $\beta_{1}(u) + w_{3}\gamma_{4}(u)$.

\begin{figure*}
	\centering
	\includegraphics[scale = 0.3]{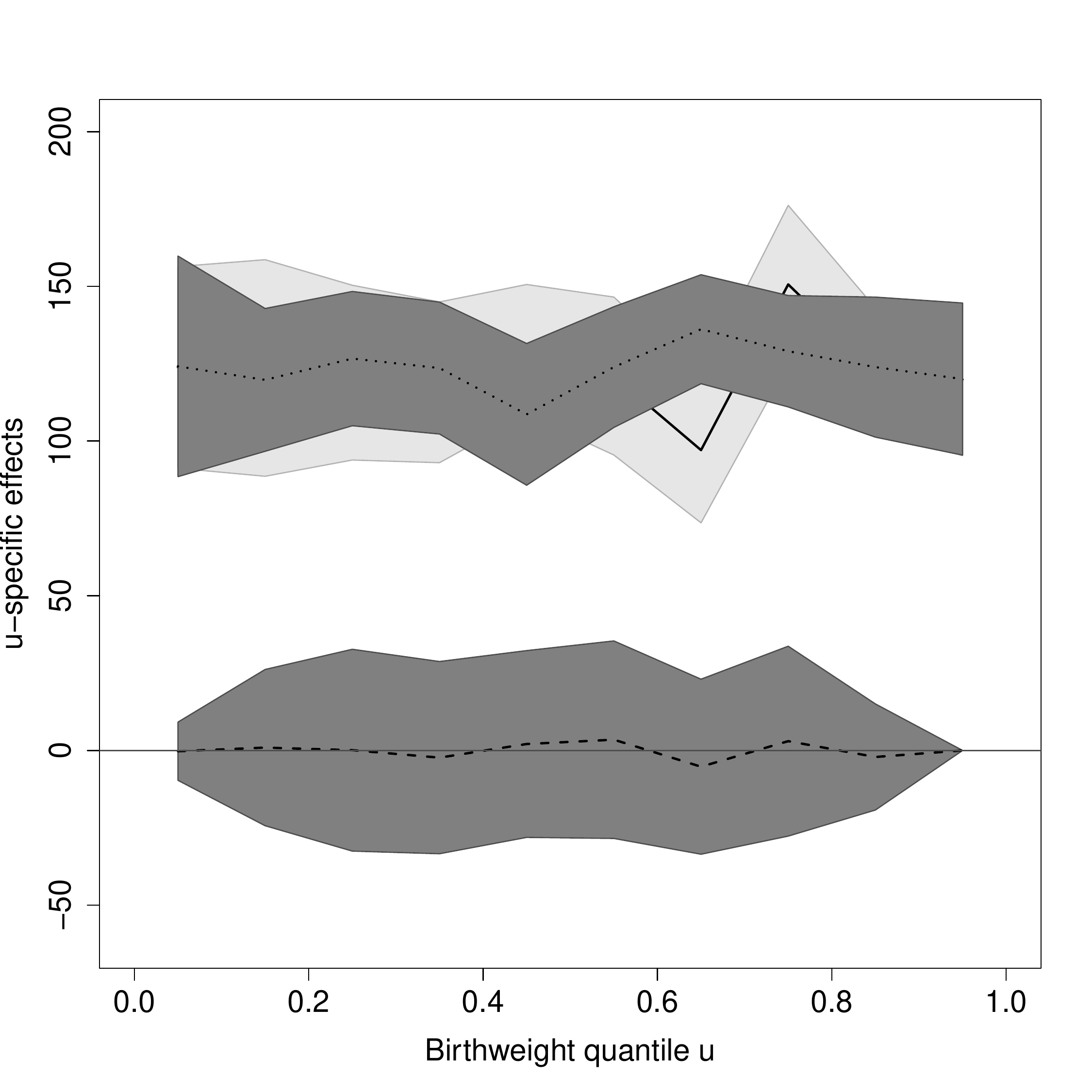}
	\caption{Sudden infant death syndrome mortality (per $100,\!000$ livebirth-years) by birthweight quantile in white singletons, National Center for Health Statistics data, United States 2001-2005. Estimates of the $u$-specific indirect effect $NIE_{x = 1|u}$ (dashed line), direct effect $NDE_{x^{\ast} = 0|u}$ (dotted line), and total effect $ACE_{u}$ (solid line) of maternal smoking on infant mortality. Shaded grey areas depict $95\%$ pointwise confidence bands.}
	\label{fig:6}
\end{figure*}

The risk of SIDS was calculated for decile intervals of birthweight, conditional on smoking, prenatal care, maternal age, alcohol consumption, and the interaction between the latter and smoking. Mortality was higher in exposed infants at all birthweight quantiles (results not shown).

The estimated $u$-specific indirect, direct, and total effects, integrated over the distribution of $W$, are shown in Figure~\ref{fig:6} for the first nine deciles of birthweight. The $u$-specific indirect effect is approximately constant across the entire birthweight distribution. The $95\%$ pointwise confidence intervals include zero at all values of $u$. As a consequence, the $u$-specific total effect is determined solely by the direct effect. These effects, too, are approximately constant across deciles. In other words, the downward shift in birthweights due to smoking, adjusted for other SIDS risk factors, has little or no bearing on the higher risk in the exposed population. The indirect, direct, and total effects, averaged over $u$, were equal to $-0.6$, $123.6$ and $123.0$, respectively. In this case the mean provides an exhaustive summary of these effects.

The estimated components of the $u$-specific indirect effects are plotted in Figure \ref{fig:7}. The estimate $\hat{r}(u)$ was approximately null at all deciles (i.e., the risk of SIDS was approximately constant across birthweight quantiles). In contrast, smoking was strongly and significantly associated with birthweight and its effect differed by birthweight quantile. Moreover, quantile effects were heterogeneous in relation to mothers' characteristics. Smoking and alcohol consumption combined, together with absence of prenatal care, determined a shift of the birthweight distribution equal to minus $1049$ grams at $u =  0.1$ and to minus $536$ grams at $u =  0.9$. In comparison, the shift caused by smoking alone was $75\%$ ($u =  0.1$) to $58\%$ ($u =  0.9$) smaller. For infants born to young women who did not receive prenatal care, smoking determined a shift of the birthweight distribution in between that provoked in the other two subpopulations.

\begin{figure*}
	\centering
	\includegraphics[width = \textwidth]{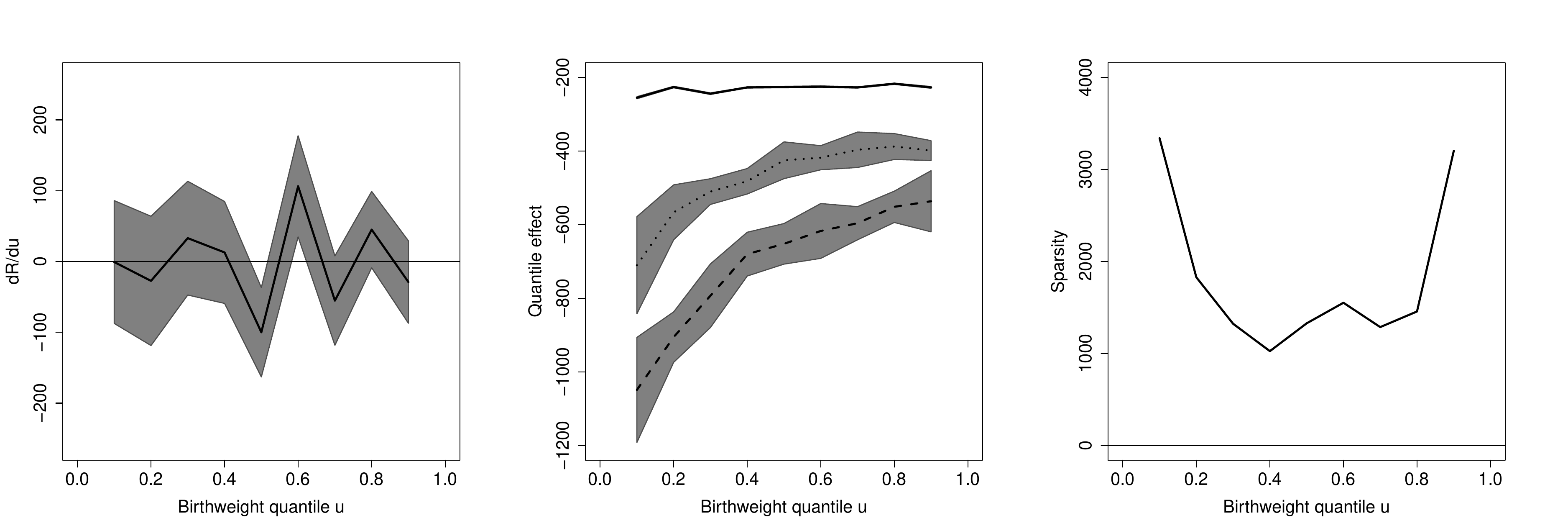}
	\caption{Sudden infant death syndrome mortality (per $100,\!000$ livebirth-years) by birthweight quantile in white singletons, National Center for Health Statistics data, United States 2001-2005. Components of the estimated $u$-specific indirect effect, with estimated quantile effect contrasting baseline (infants born to women aged 20 years or older who did not smoke or drink alcohol and who received prenatal care) and (i) infants born to women who smoked (solid line); (ii) infants born to women aged less than 20 years who smoked and who did not receive prenatal care (dotted line); (iii) infants born to women who smoked and consumed alcohol and who did not receive prenatal care (dashed line). Shaded grey areas depict $95\%$ pointwise confidence bands.}
	\label{fig:7}
\end{figure*}

\section{Final remarks}
\label{sec:6}
Expressing the outcome summary as a function of the quantiles of the mediator conditional on the exposure has two consequences: (i) firstly, the contrast between units in the unexposed and exposed populations is made on an equal footing; and (ii) secondly, under the assumption of sequential ignorability, it leads to a neat decomposition of the indirect effect which is shown to be proportional to the shift of the mediator's distribution associated with the exposure, the local density of the distribution, and the sensitivity of the risk to changes in the mediator's distribution. Complex effects are broken down into separate elements easier to understand and the relative importance of these components can be assessed. We restricted our attention to the case of an absolutely continuous mediator. Our proposed methods can be extended to the discrete case, although semi-parametric modelling of quantile functions for discrete responses is still an area that requires further research.

In our modelling approach, we discussed semi- and non-parametric methods, as these are able to reveal effects more complex than location-shift. We propose to assess the uncertainty of the estimates by means of the bootstrap, which has been also adopted in other studies focusing on quantiles of the outcome \citep{imai_2010a, shen_etal_2014}. We strongly recommend the method by \cite{kleiner_2014} when dealing with large datasets. The resampling step in the analyses of all-cause and SIDS mortality took about $100$ and $50$ minutes, respectively, on a $64$-bit operating system machine with $16$ Gb of RAM and quad-core processor at $2.93$ GHz.

Our methods have relevance for public health, in general, and child health, in particular. For example, our analysis showed that the indirect effect of smoking on infant mortality decreases steadily across the birthweight distribution, representing the entire total effect at lower quantiles of birthweight but becoming null at higher quantiles, where the direct effect prevails. Moreover, the large magnitude of the indirect effect in small infants is the result of the strong effect of smoking on lower quantiles of birthweight and the dramatic sensitivity of the mortality risk to small shifts of the birthweight distribution. However, confounding in the analysis of all-cause mortality is most certainly inevitable given the complexity and multitude of pre- and post-natal factors which are not available in the NCHS data. On the other hand, we found that increased SIDS risk due to prenatal exposure to smoking, adjusted for other birthweight-related risk factors, is not mediated by changes in the birthweight distribution. This result supports the hypothesis that smoking, which is an important birthweight determinant, does increase the risk of SIDS but not through pathways that involve birthweight. It cannot be excluded that infants exposed to smoking \textit{in utero} are exposed to smoking after birth as well \citep{shah_2006,usdept}. In contrast, LBW, which is considered a risk factor for SIDS \citep{hoffman_1988, blair_2006}, appears to be associated with an increased risk of SIDS possibly because of common prenatal factors, including smoking, that induce a spurious correlation.


\appendix
\section*{Appendix}
In this section, we prove Propositions~\ref{theo:1} and \ref{theo:2}.

\begin{proof}[Proof of Proposition~\ref{theo:1}]
Let's start by showing that under Assumption \ref{assm:1}, $\calR_{x, x^\ast}(m \mid  w)=\calR(x, m, w)$ and $R_{x,\xi^{x^\ast}_u,x^\ast}(u \mid w)=R(x,x^\ast,u,w)$.

For $m \in \R, \, x, x^\ast=0,1$, we have
\begin{align*}
\calR_{x,x^\ast}(m\mid w) &= \expect\left(Y(x,m) \mid M(x^\ast)=m, W=w\right) \\
&= \expect\left(Y(x,m) \mid X=x^\ast, M(x^\ast)=m, W=w\right) \\
&= \expect\left(Y(x,m) \mid X=x^\ast, W=w \right) \\
&= \expect\left(Y(x,m) \mid X=x, W=w \right) \\
&= \expect\left(Y(x,m) \mid X=x, M(x)=m, W=w\right) \\
&= \expect\left(Y  \mid X=x, M =m, W=w\right) = \calR(x,m, w),\\
\end{align*}
where the second and forth equalities  follow from ignorability of the treatment, the third and fifth equalities follow from ignorability of the mediator, and the last equality holds by consistency. 

Ignorability of the treatment implies that $F_{M(x^\ast) \mid W=w}(M(x^\ast)\mid W=w)= F_{M\mid X=x^\ast, W=w}(M\mid X=x^\ast, W=w)$ and
$\xi^{x^\ast}_{u} \equiv F^{-1}_{M(x^\ast)\mid W=w}(u)= F^{-1}_{M\mid X=x^\ast, W=w}(u) \equiv \xi_{u\mid x^\ast}$ for all $u \in (0,1)$. Therefore
$\calR_{x,x^\ast}(\xi_u^{x^\ast}\mid w)= \calR_{x,x^\ast}( \xi_{u\mid x^\ast}\mid w)$. Since $\calR_{x,x^\ast}(m \mid w) = \calR(x,m, w)$, for all $m \in \R$, $x, x^\ast=0,1$, we also have that $\calR_{x,x^\ast}(\xi_{u \mid x^\ast}\mid w)= \calR(x,\xi_{u\mid x^\ast}, w)$. Then,
\begin{align*}
R_{x,\xi_u^{x^\ast},x^\ast}(u \mid w) & = \calR_{x,x^\ast}(\xi_u^{x^\ast}\mid w) = \calR_{x,x^\ast}( \xi_{u\mid x^\ast}\mid w)=\calR(x,\xi_{u\mid x^\ast}, w)\\
& =  R(x,x^\ast,u,w),
\end{align*}
where the first and the last equalities follow from the identities \eqref{eq:9} and \eqref{eq:15}, respectively.

It follows that
\begin{eqnarray*}
NIE_{x|u,w} &\equiv& R_{x,\xi_u^{1},1}(u\mid w)-R_{x,\xi_u^{0},0}(u\mid w)\\
&=& R\left\{x,1,F_{M\mid X=1, W=w}\left(F^{-1}_{M\mid X=1, W=w}(u)\right),w\right\}
- \\&&  R\left\{x,0,F_{M\mid X=0, W=w}\left(F^{-1}_{M\mid X=0, W=w}(u)\right),w \right\}.
\end{eqnarray*}
\end{proof}

\begin{proof}[Proof of Proposition \ref{theo:2}]
Let $h(x^\ast; x, u, w) = \rds R(x,x^\ast,u,w)/\rds x^\ast$. Since $h$ is continuous, then according to the mean value theorem for definite integrals there exists some $\tilde{x}^\ast \in (0,1)$ such that
$$
\int_{0}^{1} h(x^\ast; x, u, w) \rds x^{\ast} = h(\tilde{x}^\ast; x, u, w).
$$
Thus we obtain
\begin{align*}
\int_{0}^{1} \left\{\frac{\rds R(x,x^\ast,u,w)}{\rds x^\ast}\right\} \rds x^{\ast} & = \frac{\rds R(x,x^\ast,u,w)}{\rds x^\ast}\bigg\rvert_{x^\ast = \tilde{x}^\ast}\\
& = R(x,1,u,w) - R(x,0,u,w),
\end{align*}
where the first equality follows from then mean value theorem and the second equality from the first fundamental theorem of calculus.
\end{proof}

By way of example, suppose $M|X,W \sim \mathcal{N}(\beta_{0} + \beta_{1} x + \beta_{2} w, \sigma(x))$, where $\beta_{0} = \beta_{2} = 0$, and $\sigma(x) = (1 + x)^{2}$. Then $F^{-1}_{M|X=x,W=w}(u) = \beta_{0}(u) + \beta_{1}(u)\cdot x$, where $\beta_{0}(u) = \Phi^{-1}(u)$ and $\beta_{1}(u) = \beta_{1} + \Phi^{-1}(u)$. Note that the slope of the quantile function of $M|X,W$ depends on $u$ since the model is heteroscedastic. Suppose also $\calR\left(x, m, w\right) = \exp(\alpha_{0} + \alpha_{1}x + \alpha_{2}m + \alpha_{3} w)$, where $\alpha_{0} = \alpha_{3} = 0$. For the sake of simplicity, we assume that there is no interaction between $W$ and $X$ or between $W$ and $M$.

Consider first $x^{\ast} \in \R$ and let
$$
R\left(x, x^\ast, u, w\right) = \exp\{\alpha_{1} x + \alpha_{2}\left(\beta_{0}(u) + \beta_{1}(u)\cdot x^\ast\right)\},
$$
Then we have
\begin{align*}
R\left(x, x^\ast, u, w\right) & = R\left(x, x^\ast, F_{M\mid X=x^\ast, W=w}\left(F^{-1}_{M\mid X=x^\ast, W=w}(u) \right), w\right)\\
& = \exp\left\{\alpha_{1} x + \alpha_{2}\left(\Phi^{-1}(\Phi\{\Phi^{-1}(u)\}) + [\beta_{1} + \Phi^{-1}(\Phi\{\Phi^{-1}(u)\})] x^{\ast} \right)\right\}
\end{align*}
as per \eqref{eq:15}, and
\begin{align*}
& \dfrac{\rds R\left\{x,x^\ast,F_{M\mid X=x^\ast, W=w}\left(F^{-1}_{M\mid X=x^\ast,W=w}(u)\right),w\right\}}{\rds x^\ast}\\
& \quad = \dfrac{\rds R(x,x^\ast, u, w)}{\rds u} \cdot  \dfrac{\rds F_{M\mid X=x^\ast, W=w}\left(F^{-1}_{M\mid X=x^\ast, W=w}(u) \right)}{\rds F^{-1}_{M\mid X=x^\ast, W=w}(u)} \cdot \dfrac{\rds  F^{-1}_{M\mid X=x^\ast, W=w}(u)  }{\rds  x^\ast}\\
& \quad = \underbrace{\frac{\alpha_{2}(1+x^{\ast})}{\phi\left(\Phi^{-1}(u)\right)}\exp\{\alpha_{1} x + \alpha_{2}\left(\beta_{0}(u) + \beta_{1}(u)\cdot x^\ast\right)\}}_{r(x,x^\ast, u, w)} \cdot \underbrace{\frac{1}{(1+x^{\ast})}\phi(\Phi^{-1}(u))}_{[s(u,x^\ast, w)]^{-1}} \cdot \underbrace{\beta_{1}(u)}_{q(u, x^\ast, w)}\\
& \quad = \alpha_{2}\beta_{1}(u)\exp\{\alpha_{1} x + \alpha_{2}\left(\beta_{0}(u) + \beta_{1}(u)\cdot x^\ast\right)\}
\end{align*}
as per \eqref{eq:17}. Note that $\alpha_{1}$ and $\alpha_{2}\beta_{1}(u)$ are, respectively, the direct and indirect effects of the exposure on the outcome.

Consider now a binary $x^{\ast}$ and take $x^\ast = 0$, $\rds x^{\ast} = 1$. The $u$-specific natural indirect effect is calculated as
\begin{align*}
NIE_{x\mid u,w} & = \exp\{\alpha_{1} x + \alpha_{2}\left(\beta_{0}(u) + \beta_{1}(u)\right)\} - \exp\{\alpha_{1} x + \alpha_{2}\beta_{0}(u)\}.
\end{align*}
According to Proposition~\ref{theo:2}, there exists $\tilde{x}^{\ast}$ such that
\begin{align*}
\alpha_{2}\beta_{1}(u)\exp\{\alpha_{1} x + \alpha_{2}\left(\beta_{0}(u) + \beta_{1}(u)\cdot x^\ast\right)\} & = \exp\{\alpha_{1} x + \alpha_{2}\left(\beta_{0}(u) + \beta_{1}(u)\right)\} \\
& \quad - \exp\{\alpha_{1} x + \alpha_{2}\beta_{0}(u)\}.
\end{align*}
It is easy to verify that the above equality is satisfied for
$$\tilde{x}^{\ast} = \frac{1}{\alpha_{2}\beta_{1}(u)}\log\left[\frac{1}{\alpha_{2}\beta_{1}(u)}(\exp\{\alpha_{2}\beta_{1}(u)\} - 1)\right],$$
for $\alpha_{2}\neq 0$, $u \in (0,1)$ and $u \neq \Phi(-\beta_{1})$.

\end{document}